\documentclass{llncs}
\usepackage{amsfonts,amssymb,amsmath,latexsym,ae,aecompl}
\usepackage{epsfig}
\usepackage{graphicx}

\newcommand{\BBB}{\vspace*{-6pt}}

\newcommand{\myparagraph}[1]{\BBB\paragraph{#1}}

\begin{document}

\title{Secure Position Verification for Wireless\\ Sensor Networks in Noisy Channels}
\author{Partha Sarathi Mandal$^1$ \and Anil K. Ghosh$^2$}
\institute{$^1$ Indian Institute of Technology, Guwahati - 781039, India.\\
$^2$ Indian Statistical Institute, Kolkata - 700108, India.\\
{\it E-mail}: psm@iitg.ernet.in, akghosh@isical.ac.in
}

\maketitle

\begin{abstract}
Position verification in wireless sensor networks (WSNs) is quite tricky in presence of attackers
(malicious sensor nodes), who try to break the verification protocol by reporting their incorrect positions (locations)
during the verification stage. In the literature of WSNs, most of the existing methods of position verification
have used trusted verifiers, which are in fact vulnerable to attacks by malicious nodes.
They also depend on some distance estimation techniques, which are not accurate in noisy channels (mediums).
In this article, we propose a secure position verification scheme for WSNs in noisy channels without relying on any trusted entities.
Our verification scheme detects and filters out all malicious nodes from the network with a very high probability.

\vspace{0.05in}
\noindent {\bf Key words:} Central limit theorems, Distributed protocol, Quantiles, Location verification, Security, Wireless networks.

\end{abstract}

\section{Introduction}


Secure position verification is important for wireless sensor networks (WSNs) because position of a sensor node is
a critical input for many WSN applications those include tracking \cite{Military},
monitoring \cite{habitat}
and geometry based routing \cite{GeoRouting}.
Most of the existing position verification protocols rely on distance
estimation techniques such as received signal strength (RSS)\cite{POS4,POS5}, time of flight (ToF)\cite{UWB}
and time difference of arrival (TDoA)\cite{POS3}.
These techniques are relatively easy to implement, but they are a little bit expensive
due to their requirement of special hardwares to estimate end-to-end distances.
These above techniques, especially RSS techniques \cite{POS4,POS5} are perfect in terms of
precision in {\it ideal} situations. The Friis transmission equation \ref{FriiEq}~\cite{LF88} used in RSS techniques
leads to this precision. But, in practice, due to the presence of noise in the network channel, signal attenuation
does not necessarily follow this equation. There are many nasty effects those have influence on both propagation time and
signal strength. So, the distance calculated using Friis equation usually differs from the actual distance. This difference, in reality, may also depend on the location of the sender and the receiver. A good position verification protocol should take care of these noises and limited precisions in distance estimation.

In this article, we use the RSS technique for position verification, where the receiving node estimates
the distance of the sender on the basis of sending and receiving signal strengths.
Here we use the term {\it node} for wireless sensor device in WSNs, which is capable of processing power and
equipped with transceivers communicating over a wireless channel. We consider that there are two types of nodes
in the system, {\it genuine} nodes and {\it malicious} nodes. While the {\it genuine} nodes follow
the implemented system functionality correctly, the {\it malicious} nodes are under the control of an adversary.
To make the verification problem most difficult, we assume that the malicious nodes know all genuine nodes
and their positions (coordinates). Once the coordinates of all genuine nodes are known, the main objective of a
malicious node is to report a suitable faking position to all these genuine nodes such that it can deceive as
many genuine nodes as possible. On the other hand, the objective of a genuine node is to detect the inconsistency
in the information provided by a malicious node. In order to do this, they compare two different estimates
of the distances, one calculated from the coordinates provided by a node and the other computed using the RSS technique. If these estimates are close, the genuine node accepts the sender as genuine, otherwise the sender node is
considered as a malicious node. Malicious nodes, however, do not go for such calculations.
They always report all genuine nodes as malicious and all malicious nodes as genuine to break
the verification protocol. In this present work, we deal with such situations and discuss how to
detect and filter out all such malicious nodes from a WSN in a noisy channel.

\myparagraph{\it Related Works:}
Most of the existing methods for secure position verification \cite{CH06,CRCS08,LP04,LCP06} rely on a fixed set
of \emph{trusted} entities (or \emph{verifiers}) and distance estimation techniques to filter out faking (malicious) nodes.
We refer to this model as the  \emph{trusted sensor} (or \emph{TS}) model. In this model, faking nodes may use some
modes of attacks that cannot be adopted by genuine nodes, such as radio signal jamming or using directional antenna
that permit to implement attacks, \emph{e.g.}, wormhole attack \cite{HPJ03,ShokriPRPH09} and Sybil attack \cite{D02}.
Lazos and Poovendran~\cite{LP04} proposed a secure range-independent localization scheme, which 
is resilient to wormhole and Sybil attacks with high probability.  Lazos \emph{et. al.}~\cite{LCP06}
further refined this scheme with multi-lateration to reduce the number of required locators, while maintaining
probabilistic guarantees. Shokri \emph{et. al.} \cite{ShokriPRPH09} proposed a secure neighbor verification protocol,
which is secure against the classic 2-end wormhole attack.
These authors assumed that there is no compromise between
external adversaries and the correct nodes or their cryptographic keys, but these adversaries control a number of
relay nodes which results in a wormhole attack.
The TS model was also considered by Capkun and Hubaux \cite{CH06} and Capkun
\emph{et. al}. \cite{CRCS08}. In \cite{CH06}, the authors presented a protocol, which relies on the distance bounding technique
proposed by Brands and Chaum \cite{DistBound}.
The protocol presented in \cite{CRCS08} relies on a set of hidden verifiers.
There are two major weakness of the TS model; firstly, it is not possible to self-organize a network in a completely
distributed way, and secondly, periodical checking is required to ensure that the trusted nodes remain trusted.
Position verification problem becomes more challenging in the case of without providing any trusted sensor nodes prior.
Dela{\"e}t \emph{et. al.}\cite{DMRT08} considerd the model as the \emph{no trusted sensor} (or \emph{NTS}) model.
Hwang \emph{et. al.}\cite{HHK08} and Dela{\"e}t \emph{et. al.}\cite{DMRT08} have investigated the verification problem with the NTS model.
In both of these articles, the
authors considered the problem, where the faking nodes operate synchronously with other nodes. The approach in \cite{HHK08}
is randomized and consists of two phases: distance measurement and filtering. In the distance measurement phase, all nodes
measure their distances from their neighbours, when faking nodes are allowed to corrupt the distance measure technique.
In this phase, each node announces one distance at a time in a round robin fashion.
Thus the message complexity is $O(n^2)$.
In the filtering phase, each genuine node randomly picks up two so-called \emph{pivot} nodes and carries out its
analysis based on those pivots.
However, these chosen pivot sensors could be malicious. So, the protocol may only give a probabilistic
guarantee. 
The approach in \cite{DMRT08} is deterministic and consists of two phases that can correctly filter out
malicious nodes, which are allowed to corrupt the distance measure technique. In the case of RSS, the protocol tolerates
at most $\lfloor \frac{n}{2} \rfloor-2$ faking sensors ($n$ being the total number of nodes in the WSN) provided no four
sensors are located on the same circle and no four sensors are co-linear. In the case of ToF,
it can handle up to $\lfloor \frac{n}{2} \rfloor - 3$ faking sensors provided no six sensors are located on
the same hyperbola and no six sensors are co-linear.
\myparagraph{\it Our results:}
The main contribution of this article is \textsc{SecureNeighborDiscovery}, a secure position verification protocol
in the NTS model in a noisy channel. To the best of our knowledge, this
is the first protocol in the NTS model in a noisy environment. The protocol guarantees
that the genuine nodes reject all incorrect positions of malicious nodes with very high probability (almost equal to $1$)
when there are sufficiently many genuine nodes in the WSN. If the noise in the network channel is negligible, this
required number of genuine nodes matches with the findings of \cite{DMRT08}, where the authors proposed a
deterministic algorithm for detecting faking sensors. However, when the noise is not negligible, each
node can only have a limited precision for distance estimation. In such cases, it is not possible to develop
a deterministic algorithm. Our protocol based on probabilistic algorithm takes care of this problem and filters
out all malicious nodes from the WSN with a
very high probability. When the number of nodes in the WSN is reasonably large, this probability
turns out to be very close to 1. So, for all practical purposes, this proposed probabilistic method
behaves almost like a deterministic algorithm. Our \textsc{SecureNeighborDiscovery} protocol can be used
to prevent Sybil attack \cite{D02} by verifying whether each message contains the real position (id)
of its sender or not. The genuine nodes never accept any message with a malicious sender location.

\section{Technical preliminaries}
\label{sec:tech}
We assume that each node knows their geographic position (coordinates) and form complete graph for
communication among themselves, i.e., each node is able to communicate with all other nodes in the WSN.
We further assume that the WSN is partially synchronous: all nodes operate in phases.
In {\it first phase}, each node is able to send exactly one message to all other nodes
without collision. Unless mentioned otherwise, we will also assume that, for each transmission,
all nodes use the same transmission power $S^s$.
Malicious nodes are allowed to transmit incorrect coordinates (incorrect identifier) to all other nodes.
We further assume that malicious nodes cooperate among themselves in an omniscient manner
(\emph{i.e.} without exchanging messages) in order to deceive the genuine nodes in the WSN. Each
malicious node obeys synchrony and transmits at most one message at the beginning of first phase and one
message at the end of it.

Let $d_{ij}$ be the true distance of node $i$ from a genuine node $j$. Since node $j$ does not know the location of node $i$, it estimates $d_{ij}$ using two different techniques, one using the RSS technique and the other using the
co-ordinates provided by node $i$. These two estimates are denoted by $\hat{d}_{ij}$ and  ${\tilde d}_{ij}$, respectively.
In the RSS technique, under {\it idealized} conditions, node $j$ can precisely measure the
distance of node $i$ using Friis transmission equation \ref{FriiEq}~\cite{LF88} given by
\begin{eqnarray}
\label{FriiEq} S^r_{ji} = S^s_{i} \left(\frac{\lambda}{4 \pi
d_{ij}}\right)^2
\end{eqnarray}
where $S^s_{i} $ is the transmission power of the sender node $i$ (here $S^s_{i}=S^s$ for all $i$),
$S^r_{ji}$ is the corresponding RSS at the receiving node $j$,
and $\lambda$ is the wave length.

If the sender node $i$ gives perfect information regarding its location (i.e., ${\tilde d}_{ij}=d_{ij}$), then the distance estimated
using the RSS technique ($\hat{d_{ij}}$) and that computed from coordinates provided by node $i$ (${\tilde d}_{ij}$) will
be equal in the ideal situation. However, in practice, when we have noise in the
channel, they cannot match exactly, but they are expected to be close. But, if
node $i$ sends an incorrect information about its location,  $|{\tilde d}_{ij}-{\hat d}_{ij}|$ can be large.

\section{RSS technique in a noisy medium}
\label{sec:rss}



The above Friis transmission equation \ref{FriiEq} is used in telecommunications engineering,
and gives the power transmitted from one antenna to another under {\it idealized} conditions.
One should note that in the presence of noise in the network, the transmission equation
may not hold, and it needs to be modified. Modifications to this equation based on the
effects of impedance mismatch, misalignment of the antenna pointing and polarization, and absorption
can be incorporated using an additional noise factor $\varepsilon$, which is supposed to
follow a Normal (Gaussian) distribution with mean $0$ and variance $\sigma^2$. The modified equation is given by
\begin{eqnarray}
\label{mFriiEq} S^r_{ji} = S^{s}_i
\left(\frac{\alpha}{d_{ij}}\right)^2+ \varepsilon_{ij}
\end{eqnarray}
Where $\varepsilon_{ij}\sim N(0, \sigma^2)$ and $\alpha =
\frac{\lambda}{4\pi}$.
However, the $\varepsilon_{ij}$s are unobserved in practice. So, the receiving node $j$ estimates the distance $d_{ij}$
using the Friis transmission equation \ref{FriiEq}, and this estimate is given by $\hat{d}_{ij} =
\alpha~ {(S^s_i/{S^r_{ji}})}^{1/2}$. Since $\varepsilon_{ij} \sim N(0,\sigma^2)$, following the $3\sigma$ limit,
$S^r_{ji}$ is expected to lie between $S^s_{i}
\left({\alpha}/{d_{ij}}\right)^2 -3\sigma$ and $S^s_{i}
\left({\alpha}/{d_{ij}}\right)^2 +3\sigma$, where $d_{ij}$ is
the unknown true distance.
Accordingly, $\hat{d}_{ij}$ is expected to lie in the range $[d_{ij} \{
1 + ({3\sigma d_{ij}^2}/{\alpha^2 S^s_{i}}) \}^{-\frac{1}{2}},
~~~~d_{ij} \{
1 - ({3\sigma d_{ij}^2}/{\alpha^2 S^s_{i}}) \}^{-\frac{1}{2}}].$ So, if the sender sends its genuine
coordinates (i.e., ${\tilde d}_{ij}={d}_{ij}$), ${\hat
d}_{ij}$ is expected to lie in the range $[{\tilde d}_{ij} \{
1 + ({3\sigma {\tilde d}_{ij}^2}/{\alpha^2 S^s_{i}}) \}^{-\frac{1}{2}}, ~~~{\tilde d}_{ij} \{
1 - ({3\sigma {\tilde d}_{ij}^2}/{\alpha^2 S^s_{i}}) \}^{-\frac{1}{2}}]$ with probability almost equal to 1
($\simeq 0.9973$). The receiver node $j$ accepts node $i$ as genuine
when $\hat d_{ij}$ lies in that range. Throughout this article, we will assume
$\sigma^2$ to be known. However, if it is unknown, one can estimate it
by sending signals from known distances and measuring the deviations in received signal strengths
from those expected in ideal situations. Looking at the distribution of these deviations, one can
also check whether the error distribution is really normal (see \cite{SW} for the test of normality
of error distributions). If it differs from normality, one can choose a suitable model for the error
distribution and find the acceptance interval using the quantiles of that distribution. For
the sake of simplicity, throughout this article, we will assume the error distribution to be normal,
which is the most common and popular choice in the statistics literature.

We assume that there are $n$ sensor nodes deployed
over a region ${\cal D}$ in a two dimensional plane, $n_0$ of them are genuine,
and the rest $n_1$ ($n_0+n_1=n$) are malicious. Though our protocol does
need $n_0$ and $n_1$ to be specified, for the better understanding
of the reader, we will use these two terms for the description and mathematical
analysis of our protocol.

\subsection{Optimal strategy for malicious sensor nodes}
\label{sec:OptMalicious}
Here, we deal with the situation, where all malicious nodes know all genuine nodes and their
positions, or in other words, they know which of the sensor nodes are genuine and which ones
are malicious. Therefore, to break the verification protocol, each malicious node
reports all genuine nodes as malicious and all malicious nodes as genuine. In addition to that,
a malicious node tries to report a suitable faking position so that it can deceive as many genuine nodes as possible.
Let $\textbf{x}_j=(x_j,y_j)~~j=1,2,\ldots,n_0$, be
the coordinates of the genuine nodes and $\textbf{x}_0=(x_0,y_0)$ be
the true location of a malicious node. Instead of reporting its original position,
the malicious node looks for a suitable faking position $\textbf{x}_f=(x_f,y_f)$ to
deceive the genuine nodes. Note that if it sends $\textbf{x}_f$ as its location,
from that given coordinates, the $j$-th ($j=1,2\ldots, n_0$) genuine node estimates its distance by
${\tilde d}_{0j}=\| \textbf{x}_j - \textbf{x}_f\|$, where $\|\cdot\|$ denotes the usual Euclidean distance.
Again, the distance estimated from the received signal is $\hat{d}_{0j} = \alpha~
({{S^s_{0}}/{S^r_{j0}}})^{1/2}$. So, the $j$-th node accepts the malicious node as genuine if ${\hat
d}_{0j}$  will lies between $ \alpha_{1,0j}={\tilde d}_{0j}
\{ 1 + (3\sigma {\tilde d}_{0j}^2/{\alpha^2
S_{s,0}}) \}^{-1/2}$ and  $ \alpha_{2,0j}={\tilde d}_{0j}
\{ 1 - (3\sigma {\tilde d}_{0j}^2/{\alpha^2
S_{s,0}}) \}^{-1/2}$. Now from equation \ref{mFriiEq}, it is easy to check that $\alpha_{1,0j} \le
{\hat d}_{0j} \le \alpha_{2,0j} \Leftrightarrow \alpha^{*}_{1,0j}= \alpha^2 S^s_{0}
[{1}/{\alpha^2_{2,0j}} - {1}/{{\hat d}^2_{0j}}
]\le \varepsilon_{0j} \le \alpha^{*}_{2,0j}=\alpha^2
S^s_{0} [{1}/{\alpha^2_{1,0j}} -
{1}/{{\hat d}^2_{0j}}]$. Let $p_{0j}^f$ be the probability that the malicious node, which is originally located at
$\textbf{x}_0$, is accepted by the $j$-th genuine node when it reports $\textbf{x}_f$ as its location.  Now, from
the above discussion, it is quite clear that $p_{0j}^f$ ($j=1,2,\ldots,n_0$) is given by
\begin{eqnarray}
\label{Fakeprob}
p_{0j}^f = P\left[{\hat d}_{0j} \in
(\alpha_{1,0j}, \alpha_{2,0j})\right]
=\displaystyle\frac{1}{\sigma\surd2\pi}
\int_{\alpha^*_{1,0j}}^{\alpha^*_{2,0j}}
exp\left(-\frac{x^2}{2{\sigma}^2}\right)dx
\end{eqnarray}
Naturally, the malicious node tries to cheat as many genuine nodes as possible. Let us define an indicator
variable $Z_{0j}^f$ that takes the value $1$ (or $0$) if the malicious nodes successfully cheats (or fails to cheat)
the $j$-th genuine node when it sends the faked location $\textbf{x}_f$.
Clearly, here $E(Z_{0j}^f) = P(Z_{0j}^f = 1) = p_{0j}^f$. So, given the coordinates of the genuine nodes ${\cal X}_0=\{{\bf x}_1,{\bf x}_2,\ldots,{\bf x}_{n_0}\}$, $\theta_{0,n_0}^{f,{\cal X}_0}= E(\sum^{n_0}_{j=1}Z_{0j}^f) =\sum^{n_0}_{j=1}p_{0j}^f$ denotes the expected number of genuine nodes to be
deceived by the malicious node if it pretends $\textbf{x}_{f}$ as its location. Naturally,
the malicious tries to find a faked position $\textbf{x}_{f}$ that maximizes $\theta_{0,n_0}^{f,{\cal X}_0}$. Let us define $\theta^{{\cal X}_0}_{0,n_0} =\sup_{{\bf x}_f \in {\cal F}_0 }\theta_{0,n_0}^{f,{\cal X}}$, where
${\cal F}_0$ is the set of all possible faking coordinates. A malicious node located at
$\textbf{x}_0$ always looks for $\textbf{x}_f \in {\cal F}_0$ such that  $\theta_{0,n_0}^{f,{\cal X}_0}= \theta_{0,n_0}^{{\cal X}_0}$.

Here one should note that the region ${\cal F}_0$ depends on the true location of the malicious node ${\bf x}_0$, and it is not
supposed to contain any point lying in a small neighborhood ${\bf x}_0$. Because in that case, ${\bf x}_0$
and ${\bf x}_f$ will be almost the same, and the malicious node will behave almost like a genuine node.
Naturally, the malicious node would not like to do that, and
it will keep the neighborhood outside ${\cal F}_0$.
The size of this neighborhood of course depends on the specific application, and the value of
$\theta_{0,n_0}^{{\cal X}_0}$ may also depend on that.

\subsection{Optimal strategy for genuine sensor nodes}
\label{sec:OptGenuine}
Let $A_0$ as the total number of nodes in the WSN that accept the malicious node located at ${\bf x}_0$
(as discussed in Section \ref{sec:OptMalicious}) as genuine. Since a malicious node is always accepted by
other malicious nodes, if there are $n_0$ genuine nodes in the WSN and  ${\cal X}_0$ denotes their co-ordinates,
for the optimum choice of the faking coordinates ${\bf x}_f$, the (conditional) expected value of $A_0$ is given by $E(A_0 \mid n_0, {{\cal X}_0})=(n-n_0)+\theta_{0,n_0}^{{\cal X}_0}$.
Now, a genuine node does not know a priori how many genuine
nodes are there is the WSN, and where they are located. So, at first, for a given $n_0$,
it computes the average of $E(A_0 \mid n_0, {{\cal X}_0})$ over all possible ${\cal X}_0$. If
${\cal D}$ denotes the deployment region (preferably a convex region) for the sensor nodes, and if
the nodes are assumed to be uniformly distributed over ${\cal D}$, this average is given by
$E(A_0 \mid n_0) = \int_{{\cal X}_0 \in {\cal D}^{n_0}} E(A_0 \mid n_0, {\cal X}_0) \psi({\cal X}_0) d{\cal X}_0$,
where $\psi$ is the uniform density function on ${\cal D}^{n_0}$. Here we have chosen $\psi$ to be uniform because it
is the most simplest one to deal with, and it is also the most common choice in the absence of any prior knowledge on
the distribution of nodes in ${\cal D}$. When we have some prior knowledge about this distribution, $\psi$
can be chosen accordingly.
Now, define $\theta_{0,n_0}=\int_{{\cal X}_0 \in {\cal D}^{n_0}}
\theta_{0,n_0}^{{\cal X}_0} \psi({\cal X}_0) d{\cal X}_0$. Clearly, $E(A_0 \mid n_0) = (n-n_0)+\theta_{0,n_0}$ depends on $n_0$,
which is unknown to the genuine node. So, it finds an upper bound for $E(A_0 \mid n_0)$ assuming that at least
half of the sensor nodes in the WSN are genuine.  Under this assumption, this upper bound is
given by $\lfloor 0.5n\rfloor + \theta_{0,\lceil0.5n\rceil}$.

\begin{theorem}
\label{Thm1_half}
If there are $n$ nodes in a WSN, and at least half of them are genuine, the expected number of acceptance
for a malicious node located at ${\bf x}_0=(x_0,y_0)$ cannot exceed $\lfloor 0.5n\rfloor + \theta_{0,\lceil0.5n\rceil}$.
\end{theorem}
\begin{proof}
Suppose there are $n_0$ genuine nodes (and $n_1=n-n_0$ malicious nodes) in the WSN, where
$n_0>\lceil n/2 \rceil$. Define ${\cal X}_0=\{{\bf x}_1,{\bf x}_2,\ldots,{\bf x}_{n_0}\}$ and
${\cal X}$ =$\{{\bf x}_1,{\bf x}_2,\ldots,$ ${\bf x}_{\lceil n/2 \rceil}\}\subset {\cal X}_0$.
Now, for given ${\cal X}_0$, the expected number of acceptance for the malicious node located at
${\bf x}_0$ is $E(A_0 \mid n_0, {\cal X}_0)$
$ = {\sup_{f \in {\cal F}_0}}\left(\sum^{n_0}_{i=1}p_{0i}^f\right) +
(n-n_0)
\le \sup_{f \in {\cal F}_0}\left(\sum^{\lceil n/2 \rceil}_{i=1}p_{0i}^f\right)  +
\sup_{f \in {\cal F}_0}\left(\sum^{n_0}_{i=\lceil n/2 \rceil+1}p_{0i}^f\right) +
(n-n_0)
\le E(A_0 \mid \lceil n/2 \rceil, {\cal X}) + (n_0-\lceil n/2 \rceil) + (n-n_0)
= \theta_{0, \lceil n/2 \rceil}^{{\cal X}} + (n-\lceil n/2 \rceil).$
Now, taking expectation w.r.t. ${\cal X}_0$, we get $E(A_0 \mid n_0) \le \theta_{0,\lceil 0.5n\rceil} + \lfloor n/2 \rfloor.$
\qed
\end{proof}

Note that $\theta_{0,\lceil 0.5n\rceil}$ and the upper bound depend on the location of the malicious node ${\bf x}_0$, So, for a genuine node, it is an unknown random quantity. Therefore, a genuine node takes a conservative
approach and computes $\theta^{*}_{\lceil n/2\rceil}= \left\lceil \sup_{{\bf x}_0 \in {\cal D}} \theta_{0,\lceil n/2\rceil}\right\rceil$.
Note that here, $\theta^{*}_{\lceil n/2\rceil}$ gives an upper bound of the expected number of genuine nodes to be deceived by a malicious node in ${\cal D}$ when there are ${\lceil n/2\rceil}$ genuine sensor nodes in the WSN.
To filter out all malicious nodes from the WSN, a genuine node follows the idea of \cite{DMRT08}. For any node,
it calculates the total number of acceptances (approvals) ($A$) and rejections (accusations) ($R$), and considers the node as malicious if $R$ exceeds
$A - {\theta^{*}_{\lceil n/2\rceil}}$. Since $A+R=n$, a node is considered to be  genuine  if
$A \ge (n +{\theta^{*}_{\lceil n/2\rceil}})/2 $. Note that if there are $n_0$ genuine nodes and $n_1$ malicious nodes in the WSN, a malicious node, on an average, can be accepted by at most $\theta^{*}_{n_0}+n_1$ nodes, and it will be rejected
by at least $n_0 - \theta^{*}_{n_0}$ nodes. So, for a malicious node $A-R-\theta^{*}_{\lceil n/2\rceil}$ is expected to be smaller than
$(n_1+2\theta^{*}_{n_0})-n_0-\theta^{*}_{\lceil n/2\rceil} \le  \lfloor n/2 \rfloor + \theta^{*}_{n_0} - n_0$ (from Theorem \ref{Thm1_half}).
Therefore, if we have $n_0 \ge \lfloor n/2 \rfloor + \theta^{*}_{n_0}$, all malicious nodes are expected to be filtered out
from the WSN. A more detailed mathematical analysis of our protocol will be given in Section \ref{sec:correctness}.
For computing $\theta^{*}_{\lceil n/2\rceil}$, a genuine node uses the statistical simulation technique \cite{Sim} assuming that the sensors are distributed over ${\cal D}$ with density $\psi$ (which is taken to be uniform in this article). First it
generates coordinates ${\bf x}_0$ for the malicious node and ${\cal X}$ for $\lceil n/2 \rceil$ genuine nodes in ${\cal D}$
to compute $\theta_{0,\lceil n/2 \rceil}^{\cal X}$ by maximizing $\theta_{0,\lceil n/2 \rceil}^{f,\cal X}$. Repeating this
over several ${\cal X}$ one gets $\theta_{0,\lceil n/2 \rceil}$ as an average of the $\theta_{0,\lceil n/2 \rceil}^{\cal X}$s.
This whole procedure is repeated for several random choices of ${\bf x}_0$ to compute  $\theta^{*}_{\lceil n/2\rceil}= \left\lceil\sup_{{\bf x}_0 \in {\cal D}} \theta_{0,\lceil n/2\rceil}\right\rceil$. Note that this is an offline calculation, and it has to be done once only.

\section{The Protocol}
\label{sec:protocol}
Based on above discussions, we develop the \textsc{SecureNeighborDiscovery} protocol.
It is a \textit{two-phase} approach to filter out malicious nodes. The \textit{first phase} is named as \textsc{AccuseApprove},
and the \textit{second phase} is named as \textsc{Filtering}.

In the \textit{first phase}, each sensor node reports its coordinates to all other nodes by transmitting
an initial message.
Next, for each pair of nodes $i$ and $j$, node $j$ computes two estimates of the distance $d_{ij}$, one using the
RSS technique ($\hat{d}_{ij}$) and the other from the reported coordinates (${\tilde d}_{ij}$), as
mentioned earlier. If $\hat{d}_{ij} \notin  (\alpha_{1,ij}, \alpha_{2,ij})$ then node $j$ accuses node $i$
for its faking position. Otherwise, node $j$ approves the location of node $i$ as genuine.
Here
$\alpha_{1,ij} = {\tilde d}_{ij} \{ 1 + ({3\sigma
{\tilde d}_{ij}^2}/{\alpha^2 S_{s,i}} )\}^{-\frac{1}{2}}$ and
$\alpha_{2,ij} = {\tilde d}_{ij} \{ 1 - ({3\sigma
{\tilde d}_{ij}^2}/{\alpha^2 S_{s,i}} )\}^{-\frac{1}{2}}$ are analogs of $\alpha_{1,0j}$ and $\alpha_{2,0j}$
defined in Section \ref{sec:OptMalicious}.
To keep track of these accusations and approvals, each node $j$ maintains an array $accus_j$, and transmits it
to all other nodes at the end of this phase.
So, in the \textit{first phase}, each node $j$ executes the \textsc{AccuseApprove} protocol which is given below.
\begin{table*}[h]
\centering
{\small
\framebox[12cm][t]{\parbox[t][][t]{12cm}{\vspace{-.5cm}
\begin{tabbing}
\=xxx\=xxx\=xxx\=xxx\=xxx\=xxx\=xxxx\=xxxx\=xxxx\=xx\=xx\=\kill
{\underline{Protocol: \textsc{AccuseApprove} (executed by node $j$)}}\\
1. \>\>$j$ exchanges coordinates by transmitting $init_j$ \& receiving $n-1$ $init_i$. \\
2. \>\>{\bf for} each received message $init_i$:\\
3. \>\>\>compute $\hat{d}_{ij}$ using the ranging (RSS) technique and \\
\>\>\> ${\tilde d}_{ij}$ using the reported coordinates of $i$.\\
4. \>\>\>   {\bf if} $\left[{\hat d}_{ij} \notin (\alpha_{1,ij}, \alpha_{2,ij})\right]$ {\bf then} $accus_j[i] \leftarrow true$ \\
   \>\>\>   {\bf else} $accus_j[i] \leftarrow false $\\
5. \>\>$j$ exchanges accusations by transmitting $accus_j$ \& receiving $n-1$ $accus_i$.
\end{tabbing}\vspace{-.45cm}
}}}

\vspace{0.05in}
\centering
{\small
\framebox[12cm][t]{\parbox[t][][t]{12cm}{\vspace{-.5cm}
\begin{tabbing}
\=xxx\=xxx\=xxx\=xxx\=xxx\=xxx\=xxxx\=xxxx\=xxxx\=xx\=xx\=\kill
{\underline{Protocol: \textsc{Filtering} (executed by node $j$)}}\\
1. \>\> $F=\phi, ~G=\{1,2,\ldots,n\}, ~n^{'} \leftarrow n$\\
2.  \>\> {\bf repeat}\{$k \leftarrow n^{'} $\\
3.  \>\>\>{\bf for} each received $accus_{i }$:  $(i \in G)$\\
4.  \>\>\>\>{\bf for} each $r: (r \in G)$  \\
5.  \>\>\>\>\>{\bf if} $accus_i[r] = true$ {\bf then}   $NumAccus_r +=1$ \\
    \>\>\>\>\>{\bf else} $NumApprove_r +=1$ \\
6.  \>\>\> $newF=\phi$.\\
7.  \>\>\>{\bf for} each sensor $i: (i \in G)$ \\
8.  \>\>\>\>{\bf if} ($NumApprove_i \ge (k+ \theta^{*}_{\lceil n/2\rceil})/2 $) {\bf then} \\
    \>\>\>\>\> $j$ considers $i$ as a genuine node.\\
    \>\>\>\> {\bf else} $j$ considers $i$ as a malicious node.\\
    \>\>\>\>\>\> filter out $i$, $newF=newF\cup\{i\}$, $n^{'} \leftarrow n^{'}-1$.\\
9.  \>\>\>  $F=F\cup newF$, $G=G\setminus newF$.\\
10. \>\>\>{\bf for} each sensor $i: (i \in newF)$ \\
11. \>\>\>\> discard $accus_i$ \& corresponding $i^{th}$ entry of $accus_r$ for all $r\in G$\\
12.  \>\> \} {\bf until($k \neq n^{'}$)}
\end{tabbing}\vspace{-.45cm}
}}}
\end{table*}

In the \textit{second phase}, each node $j$ executes the \textsc{Filtering} protocol, where it counts the number of accusations and approvals toward node $i$ including its own message.
Node $j$ finds node $i$ as malicious if the number of accusations exceeds the number of approvals
minus $\theta^{*}_{\lceil n/2\rceil}$. Conversely, node $i$ is considered as genuine if its number of approvals
is greater than or equal to $(n+{\theta^{*}_{\lceil n/2\rceil}})/2$. In this process, nodes that are detected
as malicious nodes, are filtered out from the  WSN. Next, it ignores the decisions given by these deleted nodes
and repeats the same filtering method with the remaining ones.  If there are $n^{'}$ nodes in the WSN,
a node is considered to be malicious if the number of approvals is smaller than
$(n^{'}+{\theta^{*}_{\lceil n/2\rceil}})/2$. Instead of $\theta^{*}_{\lceil n/2\rceil}$, we can use
$\theta^{*}_{\lceil {n^{'}}/2\rceil}$, but in that case, $\theta^{*}_{\lceil {n^{'}}/2\rceil}$
needs to be computed again, and it needs to be computed online. Therefore, to reduce the computing cost
of our algorithm, here we stick to $\theta^{*}_{\lceil {n}/2\rceil}$. Note that the use of
$\theta^{*}_{\lceil {n}/2\rceil}$ also makes the filtering protocol more strict in the sense
that it increases the probability of a node being filtered out. Node $j$ repeats this method until
there are no further deletions of nodes from the WSN.

The \textsc{Filtering} protocol is given above. Here $F$ and $G$ denote the set of malicious and genuine
nodes respectively. Initially, we set $F=\phi$ and $G=\{1,2,\ldots,n\}$. At each stage,
we detect some  malicious nodes and filter them out. Those nodes are deleted
from $G$ and included in  $F$. At the end of the algorithm, $G$ gives the set of nodes remaining in WSN,
which are considered to be genuine nodes. It would be ideal if the set of coordinates of the nodes in $G$ matches with
${\cal X}$. However, it might not always be possible. The main objective of our protocol is to filter out all malicious nodes
from the WSN. In the process, a few genuine nodes may also get removed. So, if not all, at the
end of the algorithm, one would like $G$ to contain most of the genuine nodes and no malicious nodes.


\section{Correctness of the protocol}
\label{sec:correctness}
To check the correctness of the above protocol, we consider the worst case scenario as mentioned before,
where all genuine nodes get accused by all malicious nodes, and each malicious node gets approved by all other malicious nodes.
Assume that there are $n_0$ genuine nodes and $n_1$ malicious nodes in the WSN. Now, for $j,j^{'}=1,2,\ldots,n_0$,  define the indicator variable $Z^{*}_{jj^{'}}=1$ if the $j^{'}$-th genuine node accepts the $j$-th genuine node,
and 0 otherwise.
So, for the $j$-th genuine node, the
number of approvals $A_j^{*}$ can be expressed as $A_j^{*}= 1 + \sum_{j^{'}=1, j^{'}\neq j}^{n_0} Z^{*}_{jj^{'}}$,
where the $Z^{*}_{jj^{'}}$s are independent and identically distributed ({\it i.i.d.}) as Bernoulli random variables with the success probability $p=P(Z^{*}_{jj^{'}}=1)=0.9973\simeq 1$. If $n_0$ is reasonably large, using the Central Limit Theorem (CLT) \cite{Feller}
for the {\it i.i.d.} case, 
one can show that (see Theorem \ref{Thm2_accept})
$P(A^{*}_j \ge (n+{\theta^{*}_{\lceil n/2 \rceil}})/2) \simeq 1 - \Phi\left(\tau \right)$, where $\Phi=$ cumulative distribution function of the standard normal distribution and
$\tau=\frac{n+ \theta^{*}_{\lceil n/2 \rceil} -2n_0p}{2\sqrt{p(1-p)(n_0-1)}}$.
Since this probability does not depend on $j$, the same expression holds for all genuine nodes.

\begin{theorem}
\label{Thm2_accept}
Assume that there are $n$ nodes in the WSN, and $n_0$ of them are genuine. If $n_0$ is sufficiently large,
for the $j$-th genuine node $(j=1,2,\ldots,n_0)$, we have the acceptance probability $P\left(A^{*}_j \ge \frac{n+\theta^{*}_{\lceil n/2 \rceil}}{2}\right) \simeq 1 -\Phi\left( \frac{n+{\theta^{*}_{\lceil n/2 \rceil}} -2n_0p}{2\sqrt{p(1-p)(n_0-1)}}\right)$.
\end{theorem}

\begin{proof}
Since all malicious nodes are assumed to be intelligent, none of them will accept the genuine node. One should also notice
that the $j$-th genuine node will always accept itself. So, for this node, it is easy to see that
$A_j^{*} -1 = \sum_{j^{'}=1,j^{'}\neq j}^{n_0} Z^{*}_{jj'}$ is the sum of ($n_0-1$) independent Bernoulli random variables, each of which takes the values $1$ and $0$
with probability $p=0.9973$ and $1-p=0.0027$, respectively. From the Central Limit Theorem (C.L.T.) for {\it i.i.d.}
random variables \cite{Feller}, we have $\sqrt{n_0-1}\left( \frac{A_j^{*}}{n_0-1} -p\right)
\sim  N\left(0, p(1-p)\right)$. Therefore, the acceptance probability of the $j$-th node is $P\left(A_j^{*} \ge \frac{n+{\theta^{*}_{\lceil n/2 \rceil}}}{2}\right)=P\left(\frac{A_j^{*}-1}{n_0-1} \ge
\frac{n+{\theta^{*}_{\lceil n/2 \rceil}}-2}{2(n_0-1)}\right) \simeq 1- \Phi \left( \frac{n+{\theta^{*}_{\lceil n/2 \rceil}}-2-2(n_0-1)p}{\sqrt{2(n_0-1)p(1-p)}}\right)\simeq 1 -\Phi\left( \frac{n+{\theta^{*}_{\lceil n/2 \rceil}} -2n_0p}{2\sqrt{p(1-p)(n_0-1)}}\right).$ 
\qed
\end{proof}

If $n+{\theta^{*}_{\lceil n/2 \rceil}} -2n_0p<0$ (equivalent to $n_0 > (n+ \theta^{*}_{\lceil n/2 \rceil} )/2$
since $p \simeq 1$), for any genuine node $j$ $(j=1,2,\ldots,n_0)$, the acceptance probability $P(A^{*}_j \ge(n+{\theta^{*}_{\lceil n/2 \rceil}})/2)$
is bigger than 1/2. Again, if $p$ is close to $1$ (which is the case here), the denominator of $\tau$ becomes
close to zero. So, in that case, the acceptance probability $P(A_j^{*} \ge (n+{\theta^{*}_{\lceil n/2 \rceil}})/2)$ turns out to be very close to $1$. Note that if we have $n_0\ge \lfloor n/2 \rfloor + \theta^{*}_{n_0}$, the condition $n_0 > (n+ \theta^{*}_{\lceil n/2 \rceil} )/2$ gets satisfied.

Now, given the coordinates of $n_0$ genuine sensor nodes ${\cal X}_0$, the malicious node, which is actually located at ${\bf x}_0$ but sends ${\bf x}_f$ as its faked location, has the number of acceptance $A_0 =n_1+\sum_{j=1}^
{n_0} Z^{f}_{0j}$, where $n_1$ is the number of malicious nodes in the WSN, and $Z^{f}_{0j} \sim B(1,p^f_{0j})$ for
$j=1,2,\ldots,n_0$ (see Section \ref{sec:OptMalicious}). Again from the discussion in Section \ref{sec:OptGenuine}, it follows that
$E(A_0) < n_1 + \theta^{*}_{n_0}$. So, if $n_0 \ge \lfloor n/2 \rfloor + \theta^{*}_{n_0}$, using Theorem \ref{Thm1_half}, it is easy
to check that $(n +\theta^{*}_{\lceil n/2 \rceil}-E(A_0)) > 0.5(n_0- \lfloor n/2 \rfloor-{\theta^{*}_{n_0}})\ge 0$, and it is expected
to increase with $n$ linearly. So, if the standard deviation of $A_0$ (square root of the variance $Var(A_0)$) remains
bounded as a function of $n$, or it diverges at a slower rate (which is usually the case), for sufficiently large number of
nodes in the WSN, the final acceptance probability of the malicious node
$P(A_0 \ge (n+{\theta^{*}_{\lceil n/2 \rceil}})/2)$ becomes very close to zero.

\begin{theorem}
\label{Thm3_final}
If we have sufficiently large number of nodes in the wireless sensor network and $n_0 \ge \lfloor n/2 \rfloor + \theta^{*}_{n_0}$, for any malicious node, the final acceptance probability $P(A_0 \ge (n+\theta^{*}_{\lceil n/2 \rceil})/2) \simeq 0$.
\end{theorem}

\begin{proof}
Define $Y$ as the number of genuine nodes in the WSN that accept the malicious node as genuine. First note that $A_0=n_1+Y$, and $Y$ can be expressed as $Y=\sum_{i=1}^{n_0} Y_i$, where the $Y_i$s are independent, and $Y_i \sim Ber(p_i)$, for the $p_i$s being the probabilities of acceptance by genuine nodes in WSN for the best choice of the faking position. Clearly, $E(Y) \leq \theta^{*}_{n_0}$, $\sigma_{n_0}^2=Var(Y)= \sum_{i=1}^{n_0}p_i(1-p_i)$
and $\rho_{n_0}^3=\sum_{i=1}^{n_0} E|Y_i - E(Y_i)|^3<\sigma_{n_0}^2$. Now, under the condition $n_0\ge \lfloor n/2 \rfloor + \theta^{*}_{n_0}$, it is easy to check that $E(A_0)< (n+\theta^{*}_{\lceil n/2 \rceil})/2$, and $(n+\theta^{*}_{\lceil n/2 \rceil})/2 - E(A_0)$ increases with $n$ linearly. So, if $\sigma_{n_0}^2=Var(Y)=Var(A_0)$ remains bounded as a function of $n$, using Chebychev's inequality or otherwise, one can show that $\lim_{n \rightarrow \infty} P(A_0 \ge (n+{\theta^{*}_{\lceil n/2 \rceil}})/2) =0$. But the most likely case is $\sigma_{n_0}^2 \rightarrow
\infty$ as $n \rightarrow \infty$. In this case, one can verify that $\rho_{n_0}/\sigma_{n_0} \rightarrow 0$ as $n \rightarrow \infty$ (or equivalently $n_0 \rightarrow \infty$). Therefore, from Liapunov's Central Limit Theorem \cite{Feller}, we have $[A_0-E(A_0)]/\sqrt{Var(A_0)} \sim N(0,1)$ and $P(A_0 \ge (n+{\theta^{*}_{\lceil n/2 \rceil}})/2) \simeq 1- \Phi\left(\frac{(n+\theta^{*}_{\lceil n/2 \rceil})/2-E(A_0)}{\sqrt{Var(A_0)}}\right).$ Now, $(n+\theta^{*}_{\lceil n/2 \rceil})/2 - E(A_0)$ grows with $n$ linearly, but
$\sqrt{Var(A_0)} \le \sqrt{n_0}/2$ grows at a slower rate. So, $P(A_0 \ge (n+{\theta^{*}_{\lceil n/2 \rceil}})/2)
\rightarrow 0$ as $n \rightarrow \infty$, and for large $n$, $P(A_0 \ge (n+\theta^{*}_{\lceil n/2 \rceil})/2) \simeq 0$. Since the result does not depend on the location of the malicious node ${\bf x}_0$, it holds for all malicious nodes present in the WSN. 
\qed
\end{proof}

Theorems \ref{Thm2_accept} and  \ref{Thm3_final} suggest that if $n$ is sufficiently large and $n_0 \ge \lfloor n/2 \rfloor + \theta^{*}_{n_0}$,
all genuine nodes in the WSN have acceptance probabilities close to 1, and all malicious nodes have acceptance
probabilities close to 0. So, it is expected that after the first round of filtering, if not all, a
large number of genuine nodes will be accepted.
On the contrary, if not all, almost all malicious nodes will get filtered out from the network. However, for  proper
functioning of the WSN, one needs to remove all malicious nodes. In order to do that, we repeat the \textsc{Filtering}
procedure again with the remaining nodes. Now, among these remaining nodes, all but a few are expected to be genuine, and because of this higher
proportion of genuine nodes, the acceptance probability of the genuine nodes are expected to increase, and those for the
malicious nodes nodes are expected to decrease further. So, if this procedure is used repeatedly, after some stage, WSN
is expected to contain genuine nodes only, and no nodes will be filtered out after that. When this is the case, our
\textsc{Filtering} algorithm stops. Note that this algorithm does not need the values of $n_0$ and $n_1$ to be specified. We need to know
$n$ only for computation of $\theta^{*}_{\lceil n/2 \rceil}$. This is the only major computation involved in our method,
but one can understand
that this is an off-line calculation. If we know a priori the values of $\theta^{*}_{\lceil n/2 \rceil}$ for different
$n$, one can use those tabulated values to avoid this computation. Note that the condition $n_0 \ge \lfloor n/2 \rfloor + \theta^{*}_{n_0}$ is only a sufficient condition under which the proposed protocol functions properly. Later, we will see
that in the presence of negligible noise (or in the absence of noise) in the WSN, this condition matches with that of \cite{DMRT08}, and in that
case, it turns out to be a necessary and sufficient condition. However, in other cases, it remains a sufficient condition only, and our protocol may work properly even when it is not satisfied. Our simulation studies in the next section will make this more clear.

\section{Simulation results}
\label{sec:simulation}
We carried out simulation studies to evaluate the performance of our proposed algorithm.
In the first part of the simulation, we calculated the value of  $\theta^{*}_{\lceil n/2 \rceil}$ using the statistical simulation
technique \cite{Sim}, and using that $\theta^{*}_{\lceil n/2 \rceil}$, in the second part, we filtered out all suspected malicious nodes from
the WSN. While maximizing $\theta_{0,\lceil n/2 \rceil}^{f,{\cal X}}$ w.r.t. ${\bf x}_f$, in order to ensure that ${\bf x}_f$ and ${\bf x}_0$ are not close, an open ball
around ${\bf x}_0$ is kept outside the search region ${\cal F}_0$.
Unless mentioned otherwise, we carried out our experiments with 100 sensors nodes, but for varying choice of $n_0$
and $n_1$ and also for different levels of noise (i.e., different values of $\sigma^2$). For choosing the value
of $\sigma^2$, first we considered two imaginary nodes (the sender and the receiver nodes) located at two
extreme corners of ${\cal D}$ and calculated the received signal strength $S^{r}_{extreme}$ for that set up
under ideal condition (see  Friis equation 1). The error standard deviation $\sigma$ was taken as smaller
than or equal to $SS=S^{r}_{extreme}/3$ to ensure that all received signal strengths remain positive (after error contamination)
with probability almost equal to $1$.


\subsection{WSN with insignificant noise ($\sigma=10^{-6}SS$)}

In this case, we observe that the value of $\theta_{0,\lceil n/2 \rceil}^{{\cal X}}$ remains almost constant and equal to $2p=1.9946\simeq2$ for varying choices of ${\bf x}_0$ and ${\cal X}$. So, we have $\theta_{\lceil n/2\rceil}^{*}=2$. In fact, in this case, $\theta^{*}_k$ turns out to be $2$ for all $k \ge 2$. So, if we choose $n_0=52$ and $n_1=48$, the condition $n_0 \ge \lfloor n/2 \rfloor + \theta_{n_0}$ gets satisfied, and one should expect the protocol to work well.
When we carried out experiment, each of the 48 malicious nodes could deceive exactly two genuine nodes, and as a result, the
number of approvals turned out to 50. So, all of them failed to reach the threshold $(n+\theta^{*}_{\lceil n/2\rceil})/2=51$,
and they were filtered out from the WSN at the very first round. On the contrary, all 52 genuine nodes had number of approvals
bigger than (47 out 52 nodes) or equal to (5 out of 52 nodes) 51, and none of them were filtered out. So, at the beginning
of the second round of filtering, we had 52 nodes in the WSN, and all of them were genuine. Since the number of
approvals for each genuine node remained the same as it was in the first round, it was well above the updated
threshold (52+2)/2=27. So, no other nodes were filtered out, and our algorithm stopped with all genuine nodes and no malicious nodes in the network. Needless
to mention that the proposed protocol led to the same result for all higher values of $n_0$. But it did not work properly
when we took $n_0=51$ and $n_1=49$. In that case, all malicious nodes had 51 approvals, and those for the genuine nodes
were smaller than or equal to 51. So, no malicious nodes but some genuine nodes were deleted at the first round of filtering. As a result, the number of approvals for the genuine nodes became smaller at the second round, and that led to the removal of those nodes from the WSN. Note that in this case, the condition $n_0 \ge \lfloor n/2 \rfloor + \theta_{n_0}$ does not get satisfied. So, here the condition is not only sufficient, but it turns out to be necessary as well.

We carried out our experiment also with 101 nodes. When there were 51 genuine and 50 malicious nodes in the WSN,
the protocol did not work properly. But in the case of $n_0=52$ and $n_1=49$, it could filter out all malicious nodes. In that case, each malicious node had 51 approvals, smaller than the threshold $(n+\theta^{*}_{\lceil n/2\rceil})/2=51.5$. But,
48 out of 52 genuine nodes were accepted by all 52 genuine nodes. So, at the end of first round of filtering, in the WSN, we had 48 genuine nodes only. Naturally, no other nodes were removed at the second round. Again this shows that
$n_0 \ge \lfloor n/2 \rfloor + \theta_{n_0}$ is a necessary and sufficient condition for the protocol to work when the noise
is negligible. This is consistent with the findings of \cite{DMRT08}, where the authors allowed no noise in the
network.

\subsection{WSN with significant  noise ($\sigma=SS$)}

Unlike the previous case, here $\theta_{0, \lceil n/2 \rceil}^{{\cal X}}$ did not remain constant for different choices of ${\bf x}_0$ and ${\cal X}$. Considering $n=100$, we computed $\theta_{0, \lceil n/2 \rceil}^{{\cal X}}$ over 500 simulations, and they ranged
between 5.9831 and 23.6964 leading to $\theta_{\lceil n/2 \rceil}^{*}=24$ and $(n+\theta_{\lceil n/2 \rceil}^{*})/2=62$. Clearly, if we start with less than 62 genuine nodes, the protocol fails as all genuine nodes
get deleted at the first round of filtering. So, we started with 62 genuine and 38 malicious nodes. One can notice that here $n_0< \lceil n/2 \rceil + \theta_{\lceil n/2 \rceil}^{*}$, and the condition $n_0 \ge \lceil n/2 \rceil + \theta_{n_0}^{*}$ does not get satisfied. But our protocol worked nicely and filtered out all malicious nodes from the WSN. This shows that the above condition is only sufficient in this case. At the first round of filtering, 54 out of the 62 genuine nodes, and 5 out
of 38 malicious nodes could reach the threshold. So, at the beginning of the second round, we had only 59 nodes in the
network leading to a threshold of (59+24)/2=41.5. Naturally, none of the malicious nodes and all the genuine nodes could
cross this threshold, and at the end of the second round of filtering, we had only 54 nodes in the WSN, all of which were genuine. As expected, no nodes were filtered out at the third round, and our algorithm terminated with 54 genuine nodes.



\subsection{A modified filtering algorithm based on quantiles}

Note that in the previous problem, if we start with 60 genuine nodes and 40 malicious nodes, the
protocol fails as all genuine nodes get deleted at the first round of filtering. Here we propose a slightly
modified version of our protocol that works even when $n_0$ is smaller than $(n+\theta_{\lceil n/2 \rceil}^{*})/2$. Instead of using $(n+\theta_{\lceil n/2 \rceil}^{*})/2$,
we use a sequence of thresholds based on different quantiles of $\theta^{\cal X}_{0,\lceil n/2 \rceil}$. At first, we begin
with the threshold $n/2$ (i.e. replace $\theta_{\lceil n/2 \rceil}^{*}$ by 0) and follow the protocol
described in Section \ref{sec:protocol}. In the process, some nodes may get
filtered out. If there are $n^{(1)}$ nodes remaining in the WSN, we use the threshold
$(n^{(1)}+\theta_{\lceil n/2 \rceil}^{0.1})/2$ (i.e. replace $\theta_{\lceil n/2 \rceil}^{*}$ by $\theta_{\lceil n/2 \rceil}^{0.1}$) and apply the filtering phase of the protocol \textsc{Filtering}
on the remaining nodes. Here $\theta_{\lceil n/2 \rceil}^{q}$
denotes the $q$-th ($0<q<1$) quantile of $\theta^{\cal X}_{0,\lceil n/2 \rceil}$, and this can be estimated from the 500
values of $\theta^{\cal X}_{0,n/2}$ observed during simulation. This procedure is repeated
with thresholds $(n^{(i)}+\theta_{\lceil n/2 \rceil}^{i/10})/2$ for $i=2,3,\ldots,9$, and finally we use the threshold $(n^{(10)}+\theta_{\lceil n/2 \rceil}^{*})/2$. The nodes remaining in the WSN after these 11 steps of filtering are considered as genuine nodes.
This algorithm worked well in our case, and it filtered out all malicious nodes from the WSN
without losing a single genuine node. In fact, all malicious nodes were filtered out after the first two steps, and there
were no deletions of nodes after that. The results for the first two steps are shown in Table \ref{table1} (in our case,
$\theta_{\lceil n/2 \rceil}^{0.1}$ was 8.6786). The total number of approvals for the deleted nodes are also reported in the table for better understanding of the algorithm.

This modified version could filter out up to 44 malicious nodes. In the case of $n_0=56$ and $n_1=44$, only one genuine node was deleted from the WSN before all malicious nodes were filtered out. However, in the case of $n_0=55, n_1=45$ our algorithm failed. In that case, all genuine nodes had 54 or 55 approvals, but almost all malicious nodes had more than 55 approvals. So, our protocol could remove only 9 malicious nodes before all genuine were filtered out.

\begin{table*}[h]
\vspace{-0.4cm}
\centering
{\small
  \caption{First two steps of filtering (based on quantiles) with $n_0=60$ and $n_1=40$.}
\begin{tabular}{ccccccc} \hline
Step($i$) & \multicolumn{2}{c}{Total nodes ($n^{(i)}$)} & Threshold& \multicolumn{2}{c}{Nodes deleted} & No. of approvals\\ \cline{2-3}
\cline{5-6}
&  Genuine & Malicious  & &  Genuine & Malicious & for deleted nodes \\ \hline
 0 & 60 & 40 & 50.00 & 0 & 1 & < 50\\
   & 60 & 39 & 49.50 & 0 & 0 &  --- \\ \hline
 1 & 60 & 39 & 53.84 & 0 & 3 &  51-54 \\
   & 60 & 36 & 52.34 & 0 & 5 & 55-56\\
   & 60 & 31 & 49.84 & 0 & 5 &57-58\\
   & 60 & 26 & 47.34 & 0 & 17 &59-61\\
   & 60 & 9 & 38.84 &0 & 9 &62-69\\
   & 60 & 0 & 34.34 & 0 &0 & ---\\
   \hline
\end{tabular}
\label{table1}}
\end{table*}
\vspace{-0.0cm}



\section{Possible improvements}
In this article, we have used the modified version of Friis transmission equation \ref{mFriiEq} for developing our
\textsc{SecureNeighborDiscovery} protocol. However, sometimes one needs empirical adjustments to the basic Friis
equation \ref{FriiEq} using larger exponents. These are used in terrestrial models, where reflected signals can lead to destructive interference, and foliage and atmospheric gases contribute to signal attenuation \cite{Fette}.
There one can consider ${S_{ji}^{r}}/{S_i^s}$ to be proportional to $G_rG_s ({\lambda}/{d_{ij}})^m$, where
$G_r$ and $G_s$ are mean effective gain of the antennas and $m$ is a scaler
typically lies in the range [2,~4]. If $m$ is known, one can develop a verification
scheme following the method described in this article. Even if it is not known,
it can be estimated by sending signals from known distances and measuring the
received signal strengths.

However, our proposed protocol is not above all limitations. In this article, we have assumed that the underlying network topology is a complete graph. But, in practice, this may not always be the case. In multi-hop network topology, our \textsc{SecureNeighborDiscovery}
protocol based on voting can be used in the neighborhood of each node, provided there are sufficiently many genuine
nodes in the neighborhood. However, the performance of this verification protocol in the case of multi-hop
network topology needs to be thoroughly investigated.

\section{Concluding remarks}

In this article, we have proposed a distributed secure position verification protocol for WSNs in noisy channels.
In this approach, without relying on any trusted sensor nodes, all genuine nodes detect the existence
of malicious nodes and filter them out with a very high probability. The proposed method is conceptually quite
 simple, and it is easy to implement if $\theta^{*}_{\lceil n/2 \rceil}$ is known. Calculation of $\theta^{*}_{\lceil n/2 \rceil}$
 is the only major computation involved in our method, but one should note that this is an off-line calculation.

In the case of negligible noise in the WSN, we have seen that the performance of our protocol matches with that of the deterministic methods of \cite{DMRT08}. However, when the noise is not negligible, each of the sensor nodes can only have a limited precision for distance estimation. In such cases, it is not possible to develop a deterministic algorithm \cite{DMRT08}. Our protocol
based on probabilistic algorithm takes care of this problem, and it filters out all malicious nodes with
very high probability. When the number of nodes in the WSN is reasonable large, this probability
turns out to be very close to 1. So, for all practical purposes, our proposed method
behaves almost like a deterministic algorithm as we have seen in Section \ref{sec:simulation}. Since the influence of noise on signal
propagation is very common in WSNs, this probabilistic approach is very practical for the
implementation perspective in the real world.

One should also notice that compared to the randomized protocol of Hwang \emph{et al.} \cite{HHK08}, our protocol
leads to substantial savings on the time and the power used for transmissions. In \cite{HHK08}, the message complexity is $O(n^2)$, since each sensor announces one distance at a time in a round robin fashion. But, in the case of our proposed protocol, $O(n)$ messages are transmitted in the first phase, and each sensor announces all distances through a single message.



\bibliographystyle{abbrv}
\small
\bibliography{mybib}

\begin{thebibliography}{10}

\bibitem{POS4}
P.~Bahl and V.~N. Padmanabhan.
\newblock {RADAR:} an in-building {RF-}based user location and tracking system.
\newblock In {\em INFOCOM}, volume~2, pages 775--784. IEEE, 2000.

\bibitem{DistBound}
S.~Brands and D.~Chaum.
\newblock Distance-bounding protocols.
\newblock In {\em EUROCRYPT'93}.

\bibitem{Sim}
P.~Bratley, B.~L. Fox, and L.~E. Schrage.
\newblock {\em A Guide to Simulation}.
\newblock Springer, 1987.

\bibitem{CH06}
S.~Capkun and J.~Hubaux.
\newblock Secure positioning in wireless networks.
\newblock {\em {IEEE} Journal on Selected Areas in Comm.}, 24(2):221--232,
  2006.

\bibitem{CRCS08}
S.~Capkun, K.~Rasmussen, M.~\v{C}agalj, and M.~Srivastava.
\newblock Secure location verification with hidden and mobile base stations.
\newblock {\em IEEE TMC}, 7(4):470--483, 2008.

\bibitem{DMRT08}
S.~Dela{\"e}t, P.~S. Mandal, M.~A. Rokicki, and S.~Tixeuil.
\newblock Deterministic secure positioning in wireless sensor networks.
\newblock In {\em DCOSS'08}, volume 5067 of {\em LNCS}, pages 469--477.
  Springer, 2008.

\bibitem{D02}
J.~R. Douceur.
\newblock The sybil attack.
\newblock In {\em IPTPS '01: Int. Workshop on Peer-to-Peer Systems}, volume
  2429 of {\em LNCS}, pages 251--260, London, UK, 2002. Springer-Verlag.

\bibitem{Feller}
W.~Feller.
\newblock {\em An Intro. to Probability Th. and Its Applications, Vol. II}.
\newblock Wiley, 1966.

\bibitem{Fette}
B.~Fette.
\newblock {\em Cognitive Radio Technology, Second Edition}.
\newblock Academic Press, 2009.

\bibitem{UWB}
R.~J. Fontana, E.~Richley, and J.~Barney.
\newblock Commercialization of an ultra wideband precision asset location
  system.
\newblock In {\em Ultra Wideband Systems and Technologies, 2003 IEEE
  Conference}, pages 369--373, 2003.

\bibitem{Military}
T.~He, S.~Krishnamurthy, J.~A. Stankovic, T.~Abdelzaher, L.~Luo, R.~Stoleru,
  T.~Yan, L.~Gu, J.~Hui, and B.~Krogh.
\newblock An energy-efficient surveillance system using wireless sensor
  networks.
\newblock In {\em MobiSys '04}, pages 270--283, 2004.

\bibitem{POS5}
J.~Hightower, R.~Want, and G.~Borriello.
\newblock {SpotON}: An indoor {3D} location sensing technology based on {RF}
  signal strength.
\newblock Technical Report UW CSE 00-02-02, Univ. of Washington, Dept. CSE,
  Seattle, WA, Feb 2000.

\bibitem{HPJ03}
Y.~Hu, A.~Perrig, and D.~B. Johnson.
\newblock Packet leashes: A defense against wormhole attacks in wireless
  networks.
\newblock In {\em INFOCOM}. IEEE, 2003.

\bibitem{HHK08}
J.~Hwang, T.~He, and Y.~Kim.
\newblock Secure localization with phantom node detection.
\newblock {\em Ad Hoc Networks}, 6(7):1031--1050, 2008.

\bibitem{GeoRouting}
B.~Karp and H.~T. Kung.
\newblock {GPSR}: greedy perimeter stateless routing for wireless networks.
\newblock In {\em MobiCom'00}, pages 243--254, New York, USA, 2000. ACM Press.

\bibitem{LP04}
L.~Lazos and R.~Poovendran.
\newblock {SeRLoc}: Robust localization for wireless sensor networks.
\newblock {\em ACM Trans. Sen. Netw.}, 1(1):73--100, 2005.

\bibitem{LCP06}
L.~Lazos, R.~Poovendran, and S.~Capkun.
\newblock {ROPE}: Robust position estimation in wireless sensor networks.
\newblock In {\em IPSN}, pages 324--331. IEEE, 2005.

\bibitem{LF88}
C.~H. Liu and D.~J. Fang.
\newblock Propagation. in antenna handbook: Theory, applications, and design.
\newblock {\em Van Nostrand Reinhold}, Chapter 29:1--56, 1988.

\bibitem{POS3}
N.~B. Priyantha, A.~Chakraborty, and H.~Balakrishnan.
\newblock The cricket location-support system.
\newblock In {\em 6th ACM MOBICOM}, Boston, MA, August 2000. ACM.

\bibitem{SW}
S.~S. Shapiro and M.~B. Wilks.
\newblock An analysis of variance test for normality (complete samples).
\newblock {\em Bometrika}, 52(3-4):591--611, 1965.

\bibitem{ShokriPRPH09}
R.~Shokri, M.~Poturalski, G.~Ravot, P.~Papadimitratos, and J.-P. Hubaux.
\newblock A practical secure neighbor verification protocol for wireless sensor
  networks.
\newblock In D.~A. Basin, S.~Capkun, and W.~Lee, editors, {\em WiSec}, pages
  193--200. ACM, 2009.

\bibitem{habitat}
R.~Szewczyk, A.~Mainwaring, J.~Polastre, J.~Anderson, and D.~Culler.
\newblock An analysis of a large scale habitat monitoring application.
\newblock In {\em SenSys '04}, pages 214--226, 2004.

\end{thebibliography}
\end{document}